\documentclass[aps,pra,onecolumn,superscriptaddress,10pt]{revtex4-2}

\usepackage[ruled]{algorithm2e}
\usepackage[margin=3cm]{geometry}

\usepackage{amsmath, amsthm, amsfonts, amssymb, bm, bbm}
\usepackage{graphicx}
\usepackage[braket]{qcircuit}
\usepackage{hyperref, cleveref}
\hypersetup{
    colorlinks,
    linkcolor={blue},
    citecolor={blue},
    urlcolor={blue}
}

\usepackage{pgfplots}
\usepackage[normalem]{ulem}
\pgfplotsset{compat=1.15}
\usetikzlibrary{external}
\xyoption{color}
\newcommand{\dottedqw}[1][-1]{\ar @{.} [0,#1]}
\newcommand{\dottedmeter}{*=<1.8em,1.4em>{\xy ="j","j"-<.778em,.322em>;{"j"+<.778em,-.322em> \ellipse ur,_{}},"j"-<0em,.4em>;p+<.5em,.9em> **\dir{-},"j"+<2.2em,2.2em>*{},"j"-<2.2em,2.2em>*{} \endxy} \POS ="i","i"+UR;"i"+UL **\dir{-};"i"+DL **\dir{-};"i"+DR **\dir{-};"i"+UR **\dir{-},"i" \dottedqw}
\newcommand{\dottedgate}[1]{*+<.6em>{#1} \POS ="i","i"+UR;"i"+UL **\dir{-};"i"+DL **\dir{-};"i"+DR **\dir{-};"i"+UR **\dir{-},"i" \dottedqw}

\newcommand{\dottedctrl}[1]{\control \qwx[#1] \dottedqw}

\newcommand{\cgate}[1]{*+<.6em>{#1} \POS ="i","i"+UR;"i"+UL **\dir{-};"i"+DL **\dir{-};"i"+DR **\dir{-};"i"+UR **\dir{-},"i" \cw}

\theoremstyle{plain}
\newtheorem{theorem}{Theorem}

\theoremstyle{definition}
\newtheorem{definition}{Definition}

\newcommand{\bb}[1]{\ensuremath{\mathbb{#1}}}
\newcommand{\mc}[1]{\ensuremath{\mathcal{#1}}}
\newcommand{\mf}[1]{\ensuremath{\mathfrak{#1}}}

\makeatletter
\def\ft#1{\def\tempa{#1}\futurelet\next\ft@i}
\def\ft@i{\ifx\next\bgroup\expandafter\ft@ii\else\expandafter\ft@end\fi}
\def\ft@ii#1{{\ensuremath{\mf{F}_{\tempa}\!\left(#1\right)}}}
\def\ft@end{{\ensuremath{\mf{F}_{\tempa}}}}
\makeatother

\makeatletter
\def\ketbra#1{\def\tempa{#1}\futurelet\next\ketbra@i}
\def\ketbra@i{\ifx\next\bgroup\expandafter\ketbra@ii\else\expandafter\ketbra@end\fi}
\def\ketbra@ii#1{\left| \tempa \middle\rangle\!\middle\langle #1 \right|}
\def\ketbra@end{\left| \tempa \middle\rangle\!\middle\langle \tempa \right|}
\makeatother

\makeatletter
\def\braket#1{\def\tempa{#1}\futurelet\next\braket@i}
\def\braket@i{\ifx\next\bgroup\expandafter\braket@ii\else\expandafter\braket@end\fi}
\def\braket@ii#1{\left\langle \tempa \middle| #1 \right\rangle}
\def\braket@end{\left\langle \tempa \middle| \tempa \right\rangle}
\makeatother

\makeatletter
\def\cketbra#1{\def\tempa{#1}\futurelet\next\cketbra@i}
\def\cketbra@i{\ifx\next\bgroup\expandafter\cketbra@ii\else\expandafter\cketbra@end\fi}
\def\cketbra@ii#1{\left| \tempa \middle)\!( #1 \right|}
\def\cketbra@end{\left| \tempa \middle)\!\middle( \tempa \right|}
\makeatother

\makeatletter
\def\cbraket#1{\def\tempa{#1}\futurelet\next\cbraket@i}
\def\cbraket@i{\ifx\next\bgroup\expandafter\cbraket@ii\else\expandafter\cbraket@end\fi}
\def\cbraket@ii#1{\left( \tempa \middle| #1 \right)}
\def\cbraket@end{\left( \tempa \middle| \tempa \right)}
\makeatother

\makeatletter
\def\dketbra#1{\def\tempa{#1}\futurelet\next\dketbra@i}
\def\dketbra@i{\ifx\next\bgroup\expandafter\dketbra@ii\else\expandafter\dketbra@end\fi}
\def\dketbra@ii#1{| \tempa \rangle\!\rangle\!\langle\!\langle #1 |}
\def\dketbra@end{| \tempa \rangle\!\rangle\!\langle\!\langle \tempa |}
\makeatother

\makeatletter
\def\dbraket#1{\def\tempa{#1}\futurelet\next\dbraket@i}
\def\dbraket@i{\ifx\next\bgroup\expandafter\dbraket@ii\else\expandafter\dbraket@end\fi}
\def\dbraket@ii#1{\langle\!\langle \tempa | #1 \rangle\!\rangle}
\def\dbraket@end{\langle\!\langle \tempa | \tempa \rangle\!\rangle}
\makeatother

\def\CJ/{Choi-Jamio{\l}kowski}

\begin{document}

\title{Randomized compiling in fault-tolerant quantum computation}
\author{Stefanie J. Beale}
\affiliation{Keysight Technologies Canada, Kanata, ON K2K 2W5, Canada}
\affiliation{Institute for Quantum Computing, University of Waterloo, Waterloo, Ontario N2L 3G1, Canada}
\affiliation{Department of Physics and Astronomy, University of Waterloo, Waterloo, Ontario N2L 3G1, Canada}
\author{Joel J. Wallman}
\affiliation{Keysight Technologies Canada, Kanata, ON K2K 2W5, Canada}
\affiliation{Institute for Quantum Computing, University of Waterloo, Waterloo, Ontario N2L 3G1, Canada}
\affiliation{Department of Applied Mathematics, University of Waterloo, Waterloo, Ontario N2L 3G1, Canada}

\date{\today}

\begin{abstract}
Studies of quantum error correction (QEC) typically focus on stochastic Pauli errors because the existence of a threshold error rate below which stochastic Pauli errors can be corrected implies that there exists a threshold below which generic errors can be corrected.
However, rigorous estimates of the threshold for generic errors are typically orders of magnitude worse than the threshold for stochastic Pauli errors.
Specifically, coherent errors have a particularly harmful impact on the encoded space because they can map encoded states to superpositions of logical and error states. 
Further, coherent errors can add up and interfere over multiple rounds of error correction or between syndrome measurements, which may result in significantly worse errors than expected under a stochastic Pauli error model.
In this paper, we present an algorithm which decoheres noise at the logical level, projecting the state of the system onto a logical state with a well-defined error.
The algorithm does not significantly increase the depth of the logical circuit (and usually does not lead to any increase in depth), and applies generally to most fault-tolerant gadgets and error correction steps.
\end{abstract}

\maketitle

Studies of quantum error correcting codes often assume a noise model which is relatively easy to model and correct effectively.
In most cases, this is a probabilistic Pauli noise model on the encoded space, with ideal syndrome extraction~\cite{Darmawan2017, Fern2008, Poulin2006}, however, recently there have been many numerical studies of relatively simple coherent noise models, which result in a wide range of effective error rates~\cite{Iyer2017, Huang2019, Chamberland2017}.
In cases where noisy syndrome extraction is considered, it typically follows a model where syndromes are probabilistically misreported because this model is straightforward to simulate~\cite{Skoric2023, Fern2008, Poulin2006, Darmawan2017, Iyer2017, Huang2019, Chamberland2017}.

Quantum systems can be used for more efficient computations because they can be in coherent superpositions of states.
When written as an expansion in terms of outer products of the superposed states, these coherences are represented by cross-terms between the states, for example $\ketbra{0}{1}$ or $\ketbra{1}{0}$ have coherence between $\ket{0}$ and $\ket{1}$.
Unintended coherences can cause noise that is difficult to model and detrimental to the computation as they can interfere and build up quickly.
Ideal measurements remove coherences (that is, the cross-terms) between eigenstates of the measured operator.
For example, if we begin in a superposed state of $\ket{0}$ and $\ket{1}$, e.g.
\begin{align}
    \ketbra{+} = \frac{1}{2}\left(\ketbra{0}{0}+\ketbra{0}{1}+\ketbra{1}{0}+\ketbra{1}{1}\right)
\end{align}
and measure in the computational basis (i.e. measure the $Z$ operator), we remove the cross terms $\ketbra{0}{1}$ and $\ketbra{1}{0}$.
Similarly, averaging over powers of $Z$ removes coherences between $\ket{0}$ and $\ket{1}$ (the eigenspaces of $Z$)~\cite{Beale2023}.

In quantum error correction, a state is encoded in a subspace of a larger physical space~\cite{gottesman2009introduction}.
We refer to this subspace as the codespace, and for stabilizer codes, there are a set of subspaces isomorphic to the codespace which we refer to as cospaces.
Ideal syndrome measurement allows us to project onto individual cospaces and distinguish which cospace the state is in after measurement, and we apply a correction based on this information.
We refer to coherence between states in different spaces as coherences between those spaces.
When non-trivial coherences arise in a stabilizer code, they take the form of coherences within cospaces and coherences between cospaces.
There are three types of error which deviate from the standard assumptions about noise in QECCs that we are concerned about:
\begin{enumerate}
    \item Measurement errors other than classical bit flips;
    \item Coherence within a cospace;
    \item Coherence between cospaces.
\end{enumerate}
Much like with the example given above for coherences between $\ket{0}$ and $\ket{1}$, ideal syndrome measurement removes coherences between cospaces.
However, syndrome measurement may not be ideal, and errors can build up and interfere during the syndrome measurement process.
We therefore need to address all three of these deviations rather than letting syndrome measurement handle coherence between cospaces.
In this paper, we propose a technique called Logical Randomized Compiling (LRC), in which we address each of these problems as follows: we address measurement errors by applying the compilation method proposed in Ref.~\cite{Beale2023} to physical measurements, and a generalization thereof to logical measurements; we address coherences within cospaces by applying a generalization of RC~\cite{rcWallman2016a} to a logical (usually Pauli) twirling group; and we address coherences between cospaces by applying random stabilizers between operations (a more detailed breakdown of how each type of gadget is compiled is given in \cref{sec:overview}).
We show that averaging over the application of random powers of stabilizers removes coherences between cospaces (which are eigenspaces of the stabilizers), analogous to how averaging over the application of random powers of $Z$ operations decoheres coherences between eigenspaces of $Z$.

Previous work has looked at the impact of RC~\cite{rcWallman2016a} on the performance of QECs \cite{Beale2018,Jain2023}.
In the absence of twirling, residual coherence can be more efficiently corrected using non-Pauli logical operations~\cite{Beale2018}.
This possibility can still be used with LRC if the logical twirl is omitted.
RC has been shown to improve the logical fidelity for a wide class of errors, however, it can sometimes decrease the logical fidelity by removing contributions from the physical noise that interfere in a helpful way~\cite{Jain2023}.
Moreover, it is hard to predict whether the coherent contributions will interfere constructively or destructively.
Another study, Ref. \cite{Greene2021}, proposed inserting random stabilizers to mimic syndrome measurement (without getting a returned syndrome) and demonstrated that it can improve the logical error rate in the absence of error correction, but did not fully analyze the technique.
Notably, as shown in \cref{thm:randomStabsProject}, the method proposed in ref. \cite{Greene2021} removes coherences between cospaces but does not address coherences within cospaces.
As in Ref.~\cite{Greene2021}, we propose inserting stabilizers throughout a computation to enforce a noise model which is more predictable and less harmful, while still using syndrome extraction and error correction to correct the remaining error.
We include a mathematical analysis of the effect of applying random stabilizers, and combine this randomization with generalizations of other noise tailoring techniques~\cite{rcWallman2016a, Beale2023}.

The paper is structured as follows.
We begin by illustrating the potential for coherent errors that cannot be removed by existing techniques in \cref{sec:coherentExample}.
We then define LRC in \cref{sec:overview} and give a qualitative explanation for how it removes undesired coherences.
We next introduce the mathematical tools and notation in \cref{sec:Preliminaries}.
We then prove that LRC removes undesired coherences in \cref{sec:Proof}.
We illustrate how our technique is applied to syndrome extraction circuits in \cref{sec:SyndromeExtraction}.
We conclude with some implementation details and open questions in \cref{sec:Discussion}.

\section{Example: Coherent errors in the repetition code}\label{sec:coherentExample}

In this section, we provide an example where LRC is helpful and RC cannot be applied.
Specifically, we look at a noisy gadget implementation of a Toffoli gate for the 3-qubit bit-flip code, and show that a well-motivated noise model can produce an error which creates coherences between cospaces that can be suppressed by LRC.
RC is not applicable in this setting because the Toffoli gate is a non-Clifford gate.

Recall the 3-qubit bit-flip code, which has the codewords $\ket{\bar{0}}=\ket{000}$ and $\ket{\bar{1}}=\ket{111}$.
This code can be described by the stabilizer generators $ZZI$ and $IZZ$ and the logical Pauli operations $\bar{I}=III$, $\bar{X}=XXX$, $\bar{Y}=-YYY$, and $\bar{Z}=ZZZ$.
We denote a rotation about a Pauli $P$ by an angle $\theta$ by 
\begin{align}
    P_\theta=\exp(-i\theta P) = \cos \theta I - i \sin \theta P.
\end{align}
Conjugating a state by $P_\theta$ results in cross-terms between the initial state and the state that $P$ maps to.
In stabilizer codes, stabilizers act trivially, logical Pauli operators map within a cospace, and error operators map between cospaces.
Therefore, coherent stabilizer operations (that is, errors of the form $S_\theta=\exp(-i\theta S)$ for a stabilizer $S$) are trivial, coherent logical Pauli operations create coherences within a cospace, and coherent errors create coherences between cospaces.

The Toffoli gate is the three-qubit gate
\begin{align}
    CCX&=\sum_{i,j\in\bb{Z}_2} \ketbra{ij}\otimes X^{ij}.
\end{align}
The corresponding ideal logical gate is
\begin{align}
    \overline{CCX}&=\sum_{i,j\in\bb{Z}_2} \ketbra{\bar{i}\bar{j}}\otimes \bar{X}^{ij}.
\end{align}
In the 3-qubit repetition code, the logical Toffoli gate can be implemented by a gate gadget composed of transversal $CCX$ gates, i.e.,
\begin{align}
    \overline{CCX} = CCX_{(1,1), (2,1), (3,1)}CCX_{(1,2), (2,2), (3,2)}CCX_{(1,3), (2,3), (3,3)},
    \label{eq:ToffGadget}
\end{align}
and the subscripted indices label which qubits the Toffoli gates act on, with a tuple $(a, b)$ indicating the $b^{th}$ qubit in the $a^{th}$ code block.

We will look at an example where the conditional $X$ in one of the physical Toffoli gates in the Toffoli gadget is overrotated.
To model this overrotation, note that $X \cong X_{\pi/2}$ up to a phase.
Then we can express an overrotation of an $X$ gate by an angle $\delta$ as $X_{\pi/2+\delta}=\exp(-i(\pi/2+\delta)X)=X\cdot X_\delta$.
A Toffoli gate overrotated by an angle $\delta$ is then given by
\begin{align}
    CCX(\delta)&=\sum_{i,j\in\bb{Z}_2} \ketbra{ij}\otimes (X\cdot X_\delta)^{ij}.
\end{align}
We will let the first Toffoli gate in the gadget defined in \cref{eq:ToffGadget} be the one which undergoes an overrotation and observe the effect when this noisy gadget acts on the $\ket{\bar{1}\bar{1}\bar{1}}$ state, that is, the case where all three code blocks are in the logical $1$ state:
\begin{align}
    CCX(\delta)_{(1,1), (2,1), (3,1)}CCX_{(1,2), (2,2), (3,2)}CCX_{(1,3), (2,3), (3,3)}\ket{\bar{1}\bar{1}\bar{1}}&=\ket{\bar{1}\bar{1}}\otimes (X\cdot X_\delta)\otimes X\otimes X\ket{\bar{1}}\\
    &=\ket{\bar{1}\bar{1}}\otimes (XII)_\delta\bar{X}\ket{\bar{1}}\\
    &=\ket{\bar{1}\bar{1}}\otimes (XII)_\delta\ket{\bar{0}}
\end{align}
In the ideal implementation, i.e. $\delta=0$, the final state is $\ket{\bar{1}\bar{1}\bar{0}}$.
The error introduced by the overrotation is then $(XII)_\delta$ on the third code block.
For the three qubit code, $XII$ is an error and so $(XII)_\delta$ maps the system to a superposition of states in different cospaces.
By applying random stabilizers, LRC removes the coherence between the cospaces, ensuring that the state is in the syndrome space.
Note that due to the symmetry of the code and the gadget implementation, changing which Toffoli gate is overrotated results in different coherent errors.

\section{Introduction to Logical Randomized Compiling}\label{sec:overview}

In this section, we will give an intuitive overview of the motivation, implementation, and impact of LRC on a fault-tolerant quantum computation.
The following sections will provide a more rigorous mathematical treatment and justification for the protocol.

Any quantum computation can be written as a sequence of three types of primitive operations:
\begin{enumerate}
\item \textit{state resets};
\item \textit{unitary operations}; and
\item \textit{measurements}.
\end{enumerate}
We call such a sequence a circuit.
The goal of an encoded implementation is to encode the operations in such a way that some errors can be detected by performing \textit{syndrome measurements}, a fourth type of operation for encoded implementations that is itself a sequence of the three primitive operations applied on the physical space.
The syndrome measurement defines a set of orthogonal spaces, that we refer to as the cospaces of the code.
An encoded implementation of a quantum operation, which we refer to as a \textit{gadget}, is a sequence of primitive operations whose composition acts on the codespace in the desired way.
To implement an error-corrected circuit, each element of the circuit is replaced by the corresponding gadget. 
Additionally, error correction gadgets are added (e.g., between each gadget arising from encoding the circuit) that consist of a syndrome  measurement followed by a recovery operator to correct the effect of any errors that occurred.
Because gadgets are the building blocks of an error-corrected (possibly fault-tolerant) implementation of a circuit, we are concerned with compiling gadgets to tailor noise into a form that is easy to handle, i.e., with no unwanted coherences and only bit flip measurement errors.

Our protocol, LRC, reduces the effect of generic Markovian errors to the form that is typically considered in QEC, making them easier to analyze and allowing standard decoders designed for stochastic Pauli errors to be used effectively.
Specifically, each type of operation is compiled to suppress deviations from the desired form of noise as follows:
\begin{enumerate}
    \item To account for measurement noise, measurements are randomly compiled~\cite{Beale2023} using logical operations.
    \item To account for coherence within a cospace, logical twirling operations are applied before and after logical unitary operations.
    \item To account for coherence between cospaces, random stabilizers are applied before and after every operation.
\end{enumerate}

More concretely, the basic building blocks of fault tolerant quantum computing are modified as follows, where stabilizers and logical operations are always selected from the current code if code switching is being used~\cite{Anderson2014}.
In all diagrams in this paper, time flows left to right, horizontal solid lines denote a logical system, dotted lines denote a physical system, and double lines denote a classical system.
When the system could be either encoded or unencoded, the convention will be to use single solid lines.

\begin{enumerate}
\item \textit{State resets}\\
A random stabilizer $S$ is applied after each encoded state reset.
This ensures that the prepared state does not have coherences between cospaces.
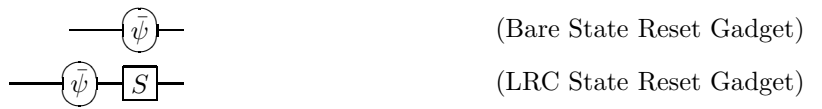
\begin{align}
    \Qcircuit @C=1em @R=.7em {
& \qw & \measure{\mbox{$\bar{\psi}$}} & \qw} \tag{Bare State Reset Gadget} \\
\Qcircuit @C=1em @R=.7em {
& \qw & \measure{\mbox{$\bar{\psi}$}} & \gate{S} & \qw} \tag{LRC State Reset Gadget}
\end{align}
\item \textit{Unitary operations}\\
A random stabilizer $S$ followed by a random logical operation $G$ drawn uniformly from a group $\bb{G}_U$ is applied before each encoded unitary operation $U$.
The compiled operation $U G^\dagger U^\dagger$ followed by a second random stabilizer $S'$ are applied after the encoded unitary operation.
The use of random stabilizers ensures that the state before and after the logical operation does not have coherences between cospaces.
The application of $G$ and $U G^\dagger U^\dagger$ performs a twirl as in RC~\cite{rcWallman2016a} which (with judicious selection of $\bb{G}_U$) removes coherence within the cospaces.
We discuss the selection of $\bb{G}_U$ in \cref{sec:Proof}.
For Clifford gates, we will typically select the logical Pauli group (or the logical Weyl group for qudits).
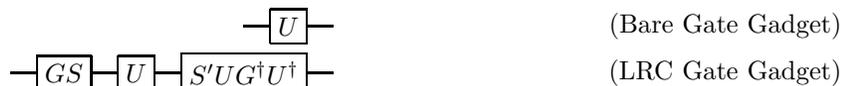
\begin{align}
    \Qcircuit @C=1em @R=.7em {
& \gate{U} & \qw} \tag{Bare Gate Gadget} \\
\Qcircuit @C=1em @R=.7em {
& \gate{G S} & \gate{U} &\gate{S' U G^\dagger U^\dagger} & \qw} \tag{LRC Gate Gadget}
\end{align}
\item \textit{Measurements}\\
Before each encoded measurement, that is each measurement of an encoded logical Weyl operation, a random stabilizer $S$ and random logical operation $L$ are performed.
The classical output of the measurement is adjusted according to the selected random logical operation.
The random stabilizer ensures that the state does not contain coherences between cospaces prior to measurement.
The inclusion of the other operations is analogous to randomized compiling for instruments~\cite{Beale2023}.

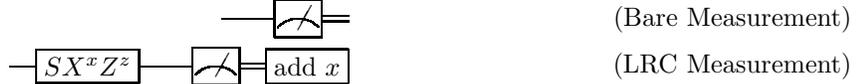
\begin{align}
    \Qcircuit @C=1em @R=.7em {
& \qw & \meter & \cw} \tag{Bare Measurement} \\
\Qcircuit @C=1em @R=.7em {
& \gate{SX^xZ^z} & \qw & \meter & \cgate{\textrm{add } x}} \tag{LRC Measurement}
\end{align}

\item \textit{Syndrome extraction}\\
A syndrome extraction measures stabilizer generators of the code by applying operations on a combined space created by the union of some readout qudits and the qudits used in the encoding and then measuring the readout qudits.
Before and after each syndrome measurement or part thereof (i.e. before and after measurement of a single bit of the syndromes), we apply a random stabilizer on the encoded space.
We have omitted some details here for ease of readability.
Specifically, the syndrome extraction circuit should be randomly compiled, and the physical measurement should undergo randomized compiling for measurements in addition to the randomization over powers of stabilizers. We will give a more complete specification in \cref{sec:SyndromeExtraction}.
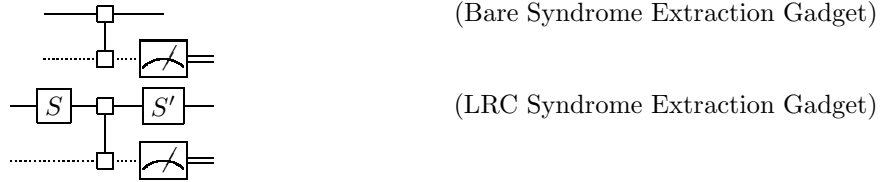
\begin{align}
\Qcircuit @C=1em @R=.7em {
& \qw & \gate{ } \qwx[1] & \qw \\
& \dottedqw & \dottedgate{} & \dottedmeter & \cw
}\tag{Bare Syndrome Extraction Gadget}\\
\Qcircuit @C=1em @R=.7em {
& \gate{S} & \gate{ } \qwx[1] & \gate{S'} & \qw\\
& \dottedqw & \dottedgate{} & \dottedmeter & \cw
}\tag{LRC Syndrome Extraction Gadget}
\end{align}
\end{enumerate}

Where possible, the additional operations introduced by LRC are compiled into adjacent operations.
As the additional operations are typically transversal single-qudit operations, very few additional physical operations will be introduced by LRC. 
These modifications apply at the logical level, however protocols for mitigating or suppressing noise at the physical level can still be used for some parts of the implementation.
In particular, dynamical decoupling or refocusing techniques and physical RC could be applied to the circuit implementations of encoded operations.

\section{Preliminaries}\label{sec:Preliminaries}

We now outline the notation and background material that we use to analyze LRC.
A single-qudit system of dimension $d$ is a system whose state space is the space $\bb{H}_{d}$ of bounded linear operators from $\bb{C}^d$ to itself.
Quantum operations are linear transformations from one quantum state space to another.
We denote by $\bb{H}_{d, e}$ the space of bounded linear operators from $\bb{H}_d$ to $\bb{H}_e$, where frequently the map is from a state space to itself, in which case $e = d$.

Throughout this paper, we denote ideal implementation maps as $\mc{A}(\cdot)$, where the argument can be a state for a state preparation, a unitary operation when the channel implements a unitary operation, or a set of projectors for a measurement.
We will reserve $\rho$ for states and capital Roman letters for unitary operations, $\Pi$ for projectors and $\bb{M}$ for projector-valued measures, that is, a set of orthogonal projectors that sum to the identity.
We denote a noisy implementation map, a noisy gadget, and an ideal gadget by replacing $\mc{A}$ by $\Theta$, $\Gamma$, and $\Gamma_{\textrm{id}}$, respectively.
We also use the short-hands
\begin{align}
    A^a &= \bigotimes_{i \in \bb{Z}_n} A^{a_i},\notag\\
    \bb{E}_{s \in S} f(s) &= \frac{1}{|S|} \sum_{s \in S} f(s),
\end{align}
where $A$ is an operator, $a \in \bb{R}^n$ is a vector, $S$ is a set and $f$ is a function whose domain is $S$.

\subsection{Weyl operators}

We now introduce the Weyl operators, which are a natural way to generalize the standard Pauli operators to a higher dimensional group.
The single-qubit Pauli operators are $\{I, X, Y, Z\}$ where $I$ is the two-by-two identity matrix,
\begin{align}\label{eq:PauliXZ}
    X &= \sum_{j \in \bb{Z}_2} \ketbra{j \oplus 1}{j}, \notag\\
    Z &= \sum_{j \in \bb{Z}_2} (-1)^j \ketbra{j},
\end{align}
and $Y = i ZX$.
We can generalize these operators to qudits by extending the sums in \cref{eq:PauliXZ} to to be over $\bb{Z}_d$ and replacing the $-1$ by a $d$th root of unity, that is
\begin{align}\label{eq:WeylXZ}
    X_d &= \sum_{j \in \bb{Z}_d} \ketbra{j \oplus 1}{j}, \notag\\
    Z_d &= \sum_{j \in \bb{Z}_d} \chi_1(j) \ketbra{j},
\end{align}
where we define
\begin{align}
    \chi_a(b) = \exp\left(\frac{2 \pi i ab}{d}\right),
\end{align}
which is a $d$th root of unity for any $a, b \in \bb{Z}_d$.
The set of single-qudit Weyl operators is then $\bb{PW}_{d,1} = \{X_d^x Z_d^z : x, z \in \bb{Z}_d\}$, which is a projective representation of $\bb{Z}_d \times \bb{Z}_d$.
For convenience, we define the functions $x, z : \bb{PW}_{d, n} \to \bb{Z}_d^n$ such that $P = X^{x(P)} Z^{z(P)}$.
The set of $n$-qudit Weyl operators is $\bb{PW}_{d, n} = \bb{PW}_{d, 1}^{\otimes n}$, which can be written as
\begin{align}
    \bb{PW}_{d, n} = \{X_d^x Z_d^z : x, z \in \bb{Z}_d^n\}.
\end{align}
Explicitly, for $x, z \in \bb{Z}_d^n$, we have
\begin{align}
    X_d^x &= \sum_{j \in \bb{Z}_d^n} \ketbra{j \oplus x}{j} \notag\\
    Z_d^z &= \sum_{j \in \bb{Z}_d^n} \chi_z(j) \ketbra{j}{j},
\end{align}
where for $a, b \in \bb{Z}_d^n$ we define
\begin{align}
    \chi_{a}(b) = \prod_{i \in \bb{Z}_n} \chi_{a_i}(b_i) = \exp\left(\frac{2 \pi i a \cdot b}{d}\right).
\end{align}
Using the relation
\begin{align}
Z^z X^x = \chi_x(z) X^x Z^z,    
\end{align}
(which follows trivially from the definitions of $X$ and $Z$), we can readily verify that the projective Weyl operators satisfy
\begin{align}\label{eq:WeylBraiding}
    PQP^\dagger = \chi_P(Q)^* Q,
\end{align}
where we define
\begin{align}
    \chi_P(Q) = \chi_{x(Q)}(z(P)) \chi^*_{x(P)}(z(Q)).
\end{align}

\subsection{Stabilizer codes}

Stabilizer codes are a subclass of quantum error correcting codes for which the encoded space is specified as the space stabilized by an Abelian group $\bb{S}$, called the stabilizer group.
That is, the codespace is the mutual $+1$ eigenspace of the stabilizers,
\begin{align}
    \{\psi : S \psi = \psi \forall S\in\bb{S}\}.
\end{align}
Typically, the stabilizers are chosen to be phase multiples of Weyl operators, and for the code space to not be empty we require $\bb{S}$ to contain no nontrivial phase multiples of the identity.
However, we will work exclusively with the operator space of the codespace, and so we define the codespace to be
\begin{align}
    \bb{O} = \{\rho : S \rho S^\dagger = \rho \forall S\in\bb{S}\}.
\end{align}
The projector onto the codespace is then $\mc{A}(\Pi_I)$ where
\begin{align}
    \Pi_I = \bb{E}_{S \in \bb{S}} S
\end{align}
is the projector onto the corresponding eigenspace.
The stabilizer group for a code which encodes $k$ logical qudits of dimension $d$ in $n$ physical qudits of dimension $d$ has order $|\bb{S}| = d^{n-k}$ and can be specified by a minimal generating set $\bb{G}$ such that $\langle\bb{G}\rangle = \bb{S}$.
The simultaneous eigenspaces of $\bb{G}$ form a set of mutually orthogonal subspaces of the physical space used to encode a logical state.
We refer to the these eigenspaces as the cospaces of the code.
There is a group $\bb{T}\subsetneq\bb{PW}_{d,n}$ called pure errors such that each pure error $T\in\bb{T}$ maps a state on one eigenspace to a state on a different eigenspace.
Thus, we define the cospace $\bb{O}_T$ of the code to be
\begin{align}
    \bb{O}_T = \{T \rho T^\dagger : \rho \in \bb{O}\}.
\end{align}
The projector onto $\bb{O}_T$ is $\mc{A}(\Pi_T)$ where
\begin{align}\label{eq:CospaceProjector}
    \Pi_T = T \Pi_I T^\dagger = \bb{E}_{S \in \bb{S}} \chi_T^*(S) S,
\end{align}
and the final equality follows from \cref{eq:WeylBraiding}.
Each stabilizer code has logical Weyl operations $\bb{L}$ for which $[L,S] = [L, T] = 0\forall L\in\bb{L}, T\in\bb{T}, S\in\bb{S}$ and which have the same commutation relations as the corresponding physical operations.

The existence of mutually orthogonal cospaces allows a QEC gadget to distinguish between sets of errors that might have occurred by measuring the stabilizer generators of the code.
Each pure error (and therefore cospace) is associated with a ditstring called the error syndrome.
We define the syndrome space to be the union of the cospaces,
\begin{align}
\cup_{T \in \bb{T}} \bb{O}_T.
\label{eq:synSpaceDef}
\end{align}
Crucially, the syndrome space is a strict subspace of the whole physical space because superpositions of states in different cospaces are not in the syndrome space.

\subsection{Gauge degree of freedom}

There exists a gauge degree of freedom in how noise is expressed on operations whereby an identity operation can be inserted between consecutive operations via any invertible matrix $M$ as $M^{-1}M$ to alter the form of the noise.
Under such a transformation, a simple circuit composed of a (possibly noisy) state preparation of a state $\rho$, a (possibly noisy) unitary operation $U$, and a (possibly noisy) measurement $\bb{M}$ is transformed as follows
\begin{align}
\Theta(\rho)\Theta(U)\Theta(\bb{M})\rightarrow\Theta(\rho)M_0^{-1} M_0 \Theta(U) M_1^{-1} M_1\Theta(\bb{M}),
\end{align}
so that we get some new implementation map $\Theta'$ such that
\begin{align}
    \Theta'(\rho) &= \Theta(\rho)M_0^{-1}\\
    \Theta'(U) & = M_0\Theta(U)M_1^{-1}\textrm{, and}\\
    \Theta'(\bb{M})&=M_1\Theta(\bb{M}).
\end{align}
As with any gauge transformation, the different implementation maps produce equivalent results, that is
\begin{align}
    \Theta(\rho)\Theta(U)\Theta(\bb{M}) =\Theta'(\rho)\Theta'(U)\Theta'(\bb{M}).
\end{align}
The matrices used to change the gauge of noises before and after each operation may be non-identical, e.g. we can have $M_0 \neq M_1$, and may, for example, depend upon the start time of an instruction.
Noisy unitaries can be factored as
\begin{align}
    \Theta(U) = \Delta_L \mc{A}(U)\Delta_R,
\end{align}
where $\Delta_L$ and $\Delta_R$ are noise maps~\cite{Wallman2018}.
Some time-dependent noise processes can be removed by an appropriate gauge transformation by folding the gauge which follows the unitary operation into the noise map preceding the next operation, thereby allowing us to assume that a noisy implementation of a unitary operation and the corresponding gadget can be written as
\begin{align}
    \Theta(U) &= \mc{A}(U)\Delta_U\textrm{, or}\\
    \Gamma(U) &= \Gamma(U)\bar{\Delta}_U
    \label{eq:unitaryNoiseFactoring}
\end{align}
for some noise maps $\Delta_U$ or $\bar{\Delta}_U$.

\section{Analysis of LRC}\label{sec:Proof}

In this section, we show how averaging over the application of random stabilizers before and after an operation results in a probabilistic map between cospaces and that twirling over the logical Weyl group enforces a probabilistic noise map within cospaces.
We then combine these results and measurement RC to show that the effective noise from the resultant protocol follows the probabilistic model we want to achieve when a logical Weyl twirl can be used, and a noise model which has fewer undesired coherences when such a twirl is not possible.
We give some examples of other twirling groups to replace the twirl over logical Weyls where needful, and note that where there is not a known twirling group for the specific operation being tailored, the application of random stabilizers is still useful to remove coherences between cospaces and can be applied regardless of the form that the surrounding operations take.

We begin with a proof that, for stabilizer codes, applying uniformly random stabilizers projects the state of the system onto the syndrome space.
This projection removes coherences between cospaces of the code.
Note that if the state begins and ends on the syndrome space, the only permitted noise is within cospaces; the effective map if we apply this averaging to project onto the syndrome space before and after an operation is a probabilistic mapping between cospaces plus possible noise within each cospace.
Because the proof is given in terms of channels and does not depend on the input state, the averaging is independent of the surrounding operations under the assumption that the stabilizer averaging occurs in the same location in the encoded circuit for each randomization.

\begin{theorem}\label{thm:randomStabsProject}
    For a code with stabilizer group $\bb{S}$, we have
    \begin{align}
        \bb{E}_{S\in\bb{S}} \mc{A}(S) = \sum_{T\in\bb{T}} \mc{A}(\Pi_T).
    \end{align}
\end{theorem}

\begin{proof}
In the natural representation~\cite{WatrousText2018}, we have
\begin{align}
    \mc{A}(M) = M^* \otimes M.
\end{align}
Substituting in \cref{eq:CospaceProjector}, we have
\begin{align*}
    \sum_{T\in\bb{T}} \mc{A}(\Pi_T)&=
    \sum_{T\in\bb{T}} \Pi_T^* \otimes \Pi_T \\
    &= \bb{E}_{S, S' \in \bb{S}} S^* \otimes S' \sum_{T\in\bb{T}} \chi^*_T(S) \chi_T(S') 
\end{align*}
Now note that as $\bb{S}$ is an Abelian group and $\chi_T$ is a character of $\bb{S}$ and thus obeys the character orthogonality relations~\cite{fulton2011representation}
\begin{align}\label{eq:CyclicOrthogonality}
    \sum_{T\in\bb{T}} \chi^*_T(S) \chi_T(S') = \delta_{S, S'} |\bb{S}|.
\end{align}
Therefore, we have
\begin{align*}
    \sum_{T\in\bb{T}} \mc{A}(\Pi_T)
    &=\bb{E}_{S \in \bb{S}} S^* \otimes S \\
    &= \bb{E}_{S\in\bb{S}} \mc{A}(S),
\end{align*}
as claimed, where the factor of $|\bb{S}|$ cancels with its inverse in the average over $S'$.
\end{proof}

By \cref{thm:randomStabsProject}, averaging over applications of stabilizers is equivalent to projecting onto the syndrome space, which removes coherences between cospaces.
Now that we have a method that handles coherence between cospaces, we next look at a method to remove or reduce coherences within cospaces.
Specifically, we look at the effect of twirling an operation over a group up to a correction that preserves the operation's ideal action.
We show that when this method is applied to a noisy implementation of an operation, it can be used to address coherences within a cospace.
A standard twirl (without correction to preserve the desired action) has the following definition:
\begin{definition}
Given a group $\bb{G}$, a channel $\Lambda$, we define the twirl of $\Lambda$ over $\bb{G}$ as

\begin{align}
    \Lambda^{\circlearrowright\bb{G}} \equiv \bb{E}_{G\in\bb{G}} \mc{A}(G^\dagger)\Lambda\mc{A}(G).
\end{align}
\end{definition}
We now show that applying a corrected twirl to a noisy implementation of a unitary operation results in an effective map comprised of the ideal unitary operation followed by a twirled version of the noise.
This is useful for the prevention of coherences within cospaces because twirled operations do not induce coherences between the states that would be returned by applications of different elements of the twirling group.
We can then remove coherences within cospaces by selecting a twirling group whose elements act within a cospace.
Note that twirling only over the logical Weyl group does not address coherences between cospaces because the representation of the logical Weyl group on the physical space is not multiplicity free.
That is, states outside of the syndrome space can persist despite a twirl over the logical Weyl group.
The following is a straightforward generalization of randomized compiling~\cite{rcWallman2016a} without assuming any structure for the twirling group.

\begin{theorem}\label{thm:twirlyWhirly}
Let $\bb{G}_U$ be a group, $U$ be a unitary operator.
Averaging over specific ideal operations applied before and after a noisy encoded unitary gadget $\Gamma(U)$ produces an effective channel comprised of the ideal action of $U$ followed by the noise introduced by $\Gamma(U)$ twirled over $\bb{G}_U$.
These operations take the form
\begin{align}
    G\tag{Before}\\
    UG^\dagger U^\dagger \tag{After}
\end{align}
for $G\in\bb{G}$.
\end{theorem}

\begin{proof}
\begin{align}
    \bb{E}_{G\in\bb{G}}\mc{A}(UG^\dagger U^\dagger)\Gamma(U)\mc{A}(G)&=\mc{A}(U)\bb{E}_{G\in\bb{G}}\mc{A}(G^\dagger)\Delta_U\mc{A}(G)\\
    &=\mc{A}(U)\Delta_U^{\circlearrowright\bb{G}},
\end{align}
where $\Delta_U \equiv \mc{A}(U^\dagger)\Gamma(U)$ can be interpreted as the noise introduced by $\Gamma(U)$, for which a noisy unitary gadget for $U$ would have an implementation map $\Gamma(U)=\mc{A}(U)\Delta_U$, consistent with \cref{eq:unitaryNoiseFactoring}.
\end{proof}
In practice, the additional operations $G$ and $UG^\dagger U^\dagger$ are compiled in to surrounding operations.
As such, we can assume that the additional operations are ideal under suitable assumptions of gate-independent noise~\cite{rcWallman2016a}.
Under gate-dependent noise, we can perturb from the average and apply \cref{thm:twirlyWhirly} provided the perturbations are small.

We now explore different options for the logical twirling group to use when applying \cref{thm:twirlyWhirly} to logical unitary operations.
Note that for the purposes of the following discussion we are omitting the random stabilizers from consideration, however, we recommend applying both techniques in combination.
A compiled version of a unitary operation $U$ of the form $(UG^\dagger U^\dagger) U G$ will always be equivalent to the bare unitary $U$, assuming that $G$ is unitary.
We can therefore in theory select any twirling group composed of unitary operations to implement the twirling component of LRC, based on what noise tailoring we wish to achieve.
In practice, however, ease of implementation matters; we generally want to select twirling operations for which the implementation of $G$ and $UG^\dagger U^\dagger$ is composed of single-qudit operations where possible so that twirling operations can be more easily compiled into neighbouring operations and introduce minimal noise where such compilation is not possible.
The suggested twirling groups, $\bb{G}_U$, given in this section assume that Weyl operations are ``easy", as is $U$.

\subsection{Clifford operations}

Clifford operations preserve the projective Weyl group, that is for any Clifford operation $C$, $\mc{A}_C(L)\in\bb{PW}_{d,n}$ $\forall L\in\bb{PW}_{d,n}$.
For any unitary operation $C$ in the Clifford group, we select $\bb{G}_C = \bb{L}$.
Then the LRC implementation of $C$ takes the form

\begin{align}
    \Qcircuit @C=1em @R=.7em {
& \gate{\Gamma(C)} & \qw} \tag{Bare Clifford Gadget} \\
\Qcircuit @C=1em @R=.7em {
& \gate{LS} & \gate{\Gamma(C)} &\gate{S'CLC^\dagger} & \qw} \tag{LRC Clifford Gadget}
\end{align}
for gates $L\in\bb{L}$.
Twirling over $\bb{L}$ results in a stochastic Weyl channel, so the noise introduced during the compiled Clifford is stochastic Weyl.

\subsection{T gates}

There are some known instances where non-Clifford gadgets can be compiled to further tailor the noise arising from the imperfect gadget implementation.
The $T$ gate is one such gate with a known twirl.
For $d=2$, we can apply operations before and after the $T$ gadget, defined as $T = \textrm{exp}(-i\pi Z/8)$, to reduce the effective noise from this gadget to a stochastic Pauli channel, with the further restriction that $p_X=p_Y$.
This compilation is possible using the dihedral group \cite{Dugas2015}, i.e. setting $\bb{G}_U = \langle \bb{PW}_{2,n}, T^2\rangle$ to perform the twirl.
Then, for any $G\in\bb{G}_U$, we can factor $G=RL$ for a rotation $R\in\langle T^2\rangle$, $L\in\bb{PW}_{2,n}$ and simplify the compilation gates to get

\begin{align}
    \Qcircuit @C=1em @R=.7em {
& \gate{\Gamma(T)} & \qw} \tag{Bare T Gadget} \\
\Qcircuit @C=1em @R=.7em {
& \gate{RLS} & \gate{\Gamma(T)} &\gate{S'LR\sqrt{Z}^{x(L)}} & \qw} \tag{LRC T Gadget}
\end{align}
where $a(L)$ for $a\in\{x,z\}$ is the power of $A$ in $L$ such that $L=X^{x(L)}Z^{z(L)}$.
For any encoding for which we can implement the Clifford group and the $T$ gate fault-tolerantly (which together form a universal gateset), we can therefore apply LRC to get a universal gateset for which the effective noise is a stochastic Weyl channel and which remains in the union of the cospaces of the code.

\subsection{Gates with no known twirl}

For operations for which there is no known useful twirling group, we can still apply the stabilizer randomization to project back onto the cospaces, thereby reducing some of the effects of coherent errors.
This amounts to setting $\bb{G}_U=\{I\}$ so that the compilation looks like

\begin{align}
    \Qcircuit @C=1em @R=.7em {
& \gate{\Gamma(U)} & \qw} \tag{Bare Gate Gadget} \\
\Qcircuit @C=1em @R=.7em {
& \gate{S} & \gate{\Gamma(U)} &\gate{S'} & \qw} \tag{LRC Gate Gadget}
\end{align}

\section{Syndrome extraction circuits}\label{sec:SyndromeExtraction}

A QEC gadget is composed of a syndrome extraction step followed by a recovery operation conditioned on the extracted syndrome.
Syndrome extraction is performed by coupling the encoded register to another register of qudits called the readout register to perform a measurement of each of the stabilizer generators.
In this section, we show how to apply LRC, focusing on syndrome extraction as it is the piece of a fault-tolerant encoding with the highest number of distinct types of operations working in tandem and so the most complicated to compile (a syndrome extraction contains unencoded measurement and state preparation as well as operations which act on an encoded and unencoded register).
Ideal syndrome extraction is a common assumption in studies of QEC because the resulting circuits are simpler and can often be efficiently simulated.
Specifically, for channels containing multiple rounds of error correction with imperfect syndrome extraction (specifically when the imperfection is not constrained to probabilistically misreported syndromes), the build up and interference of terms within and outside of the syndrome space quickly becomes complicated and expensive to track.

Under LRC, we apply random stabilizers before and after the syndrome extraction circuit to project the state of the encoded qudits onto the syndrome space.
The projection onto the syndrome space removes any coherences between cospaces so that these coherences cannot interact and build up during the syndrome extraction circuit (\cref{thm:randomStabsProject}).
We apply a corrected twirl, as in \cref{thm:twirlyWhirly}, to ensure that the effective noise for coupling operations is Pauli both on the encoded register and the readout register.
Finally, we apply measurement RC \cite{Beale2023} to the measurements of the readout qudits.

We note that in practice, a refocusing sequence and stabilizer randomization would be introduced during any idling time on the encoded register while the syndrome register is being measured.
We discuss this more in \cref{sec:idling}.
The refocusing sequence and random stabilizers applied while the readout register is being measured account for any noise that might accumulate on the encoded register during the measurement and ensure that the final state is in the syndrome space.

Though we can apply LRC to the full syndrome extraction, it is beneficial to apply it instead to the syndrome extraction performed for each individual syndrome dit.
In this way, we prevent coherences that could be removed by applying random stabilizers from persisting and interfering during subsequent operations before the full syndrome extraction is completed.
We will now step through the compilation process for a single-dit syndrome extraction, beginning with a bare syndrome extraction gadget,
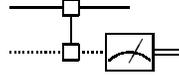
\begin{align*}
&\Qcircuit @C=1em @R=.7em {
& \qw & \gate{ } \qwx[1] & \qw \\
& \dottedqw & \dottedgate{} & \dottedmeter & \cw
}
\end{align*}
where the syndrome is extracted by writing the eigenvalue of a stabilizer generator $A$ to the readout qudit using a controlled Weyl operation, so that
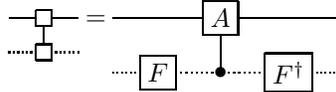
\begin{align}
\Qcircuit @C=1em @R=.7em {
& \gate{ } \qwx[1] & \qw \notag\\
& \dottedgate{} & \dottedqw}&=\Qcircuit @C=1em @R=.7em {
&\qw  & \gate{A} & \qw & \qw \\
 & \dottedgate{F} & \dottedctrl{-1} & \dottedgate{F^\dagger} & \dottedqw \\
}
\end{align}
The randomly compiled controlled Weyl operation (including the basis change implemented by $F$ gates) is then
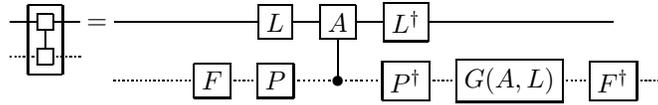
\begin{align}
\Qcircuit @C=1em @R=.7em {
& \gate{ } \qwx[1] & \qw \notag\\
& \dottedgate{} & \dottedqw\gategroup{1}{2}{2}{2}{.7em}{-}}&=\Qcircuit @C=1em @R=.7em {
 & \qw & \qw & \qw & \gate{L} &\gate{A} & \gate{L^\dagger} & \qw & \qw \\
 & \dottedqw & \dottedqw & \dottedgate{F} & \dottedgate{P}  & \dottedctrl{-1} & \dottedgate{P^\dagger} & \dottedgate{G(A,L)} & \dottedgate{F^\dagger} & \dottedqw
}
\end{align}
where the $L$ operations are drawn from the logical Weyl group, $G(A,L)$ is a Weyl operator that corrects for the operation that propagates from the $L$ acting on the encoded register to the readout register via the controlled-$A$ operation, and $P$ is an element of the Weyl group on the physical register.
Applying this random compilation to the syndrome extraction circuit, we get
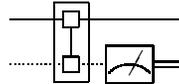
\begin{align*}
&\Qcircuit @C=1em @R=.7em {
& \qw & \gate{ } \qwx[1] & \qw & \qw \\
& \dottedqw & \dottedgate{} & \dottedmeter & \cw
\gategroup{1}{3}{2}{3}{.7em}{-}}
\end{align*}
Next, we add random stabilizers before and after the syndrome extraction circuit to get
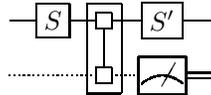
\begin{align*}
\Qcircuit @C=1em @R=.7em {
& \gate{S} & \gate{ } \qwx[1] & \gate{S'} & \qw\\
& \dottedqw & \dottedgate{} & \dottedmeter & \cw
\gategroup{1}{3}{2}{3}{.7em}{-}}
\end{align*}
This ensures that the state before and after the noisy measurement is in the syndrome space (\cref{thm:randomStabsProject}).
Finally, we add measurement RC to get
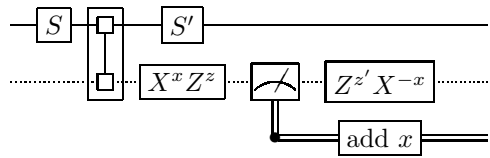
\begin{align*}
    \Qcircuit @C=1em @R=.7em {
& \gate{S} & \gate{ } \qwx[1] & \gate{S'} & \qw & \qw &\qw & \qw\\
& \dottedqw &  \dottedgate{} &\dottedgate{X^xZ^z} & \dottedmeter\cwx[1] & \dottedgate{Z^{z'}X^{-x}} & \dottedqw & \dottedqw \\
& & & & \control & \cgate{\textrm{add } x} & \cw & \cw,
\gategroup{1}{3}{2}{3}{.7em}{-}}
\end{align*}
This ensures that the measurement decouples the readout space from the encoded space and that measurement noise takes the form of probabilistically misreported syndromes.

Optionally, we can add operations to account for noise which accumulates on the encoded register during measurement (as measurement takes longer than unitary operations and the encoded register can undergo noise while idling).
We let these operations be denoted $\Phi$ and include a discussion about what form they might take in \cref{sec:idling}.
The simplest option would be to repeatedly apply random stabilizers.
Including this adjustment, the fully compiled syndrome extraction is

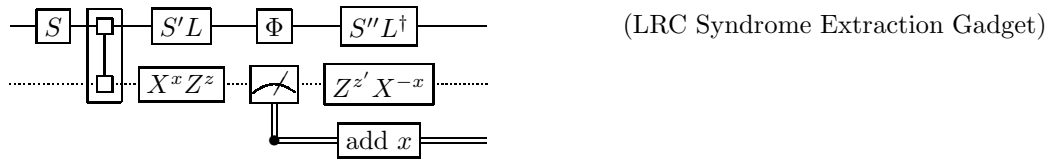
\begin{align}
    \Qcircuit @C=1em @R=.7em {
& \gate{S} & \gate{ } \qwx[1] & \gate{S'L} & \gate{\Phi} & \gate{S''L^\dagger}&\qw & \qw\\
& \dottedqw &  \dottedgate{} &\dottedgate{X^xZ^z} & \dottedmeter\cwx[1] & \dottedgate{Z^{z'}X^{-x}} & \dottedqw & \dottedqw \\
& & & & \control & \cgate{\textrm{add } x} & \cw & \cw,
\gategroup{1}{3}{2}{3}{.7em}{-}}\tag{LRC Syndrome Extraction Gadget}
\end{align}
We add a twirl of $\Phi$ over the logical Weyl operators to ensure that the noise introduced on the encoded register during idling is a stochastic Weyl channel.
Because $\Phi$ is an implementation of the identity channel, the twirled version will simply be a stochastic Weyl map.
Further, since the entire syndrome extraction is composed of Clifford operations, we can commute all of the stochastic Weyl noise terms that remain after compiling each element to one end of the circuit (e.g. the end) and combine them into a single stochastic Weyl error channel with bit flip errors on the on the readout register and reported syndrome.

Doing a complete syndrome extraction requires that we repeat this process $n-k$ times, once for each generator.
Because each repetition is a circuit that has the same form, i.e. composed of Clifford gates and a measurement, we can similarly move all of the remaining noise processes to the end of the sequence of syndrome dit extractions.
The effective noise on the encoded register after completing the syndrome extraction is then a stochastic Weyl map on the cospaces of the code, and the confusion matrix for the complete syndrome can be constructed by taking products of the elements of the confusion matrices for individual syndrome dit measurements, as in \cite{Beale2021}.
The syndrome dits may additionally be left in the incorrect state by the syndrome measurements, but this is accounted for by the reset performed on the readout registers between syndrome measurements.

\section{Discussion}\label{sec:Discussion}

In this section, we discuss some of the implementation details that are relevant to applying LRC effectively.
Specifically, we explore the considerations for compiling the operations introduced by LRC into a circuit in such a way that added noise and overhead are minimal to none in \cref{sec:overhead}, discuss some options for compiling idling registers in \cref{sec:idling}, briefly address some considerations regarding code switching in \cref{sec:codeSwitch}, and finally comment on the applicability of RC and LRC to single-shot algorithms in \cref{sec:singleShot}.

\subsection{Timing and Compilation}\label{sec:overhead}

We expect that the gates introduced by LRC will generally be transversal single-qudit gates (SQGs), which introduce less noise than other operations.
We also assume that each gadget begins and ends with a cycle of single-qudit gates (with appropriate modifications for state reset and measurement gadgets).
Under these assumptions, we can compile the gates introduced by LRC into the start and end of each gadget, and also compile the SQGs at the end of one gadget with the SQGs at the beginning of the next gadget.
LRC will therefore introduce minimal overhead in terms of gate count, and in the assumption that SQGs have gate-independent noise, will not introduce any additional noise to the system.

One concern that arises for implementing the stabilizer averaging is that each shot of a QEC circuit will be slightly different as different syndromes are measured and, accordingly, different recovery operations are applied.
Because \cref{thm:randomStabsProject} is proven using channels, the input state to the randomization does not impact the proof.
As such, the surrounding operations do not impact the effect of randomizing over stabilizers as long as the randomization occurs at the same location in the circuit for each shot.
The timing does not need to be exact in terms of wall clock time because delays are identity operations and inserting identity operations does not impact the implemented map.

An additional note is that while RC removes the coherent part of the noise much like LRC does, there are some instances where the coherent part of physical noise can improve the logical fidelity.
In such cases, removing the coherence at the physical level worsens performance.
In contrast, LRC can be used without a physical twirl, which will not remove the contribution of the physical coherent noise to the logical fidelity.

Moreover, LRC can be used in cases where RC would not be viable.
For example, stabilizer averaging can be used even in cases where the gate being implemented does not have an easily implementable twirling group, so for a gadget implementation of a non-Clifford gate like the Toffoli gate example described for the 3-qubit code in \cref{sec:coherentExample}, some coherences can be removed by LRC, while RC could not be implemented for this use-case.

\subsection{Idling}\label{sec:idling}

It is common for some subset of qudits in a computation to be idle while others undergo operations.
Perhaps most notably, measurements take longer to implement than gates so if the computation cannot proceed beyond the measurement until it has been completed, some registers may be left idle while the measurement is underway.
In \cite{Beale2023}, operations on unmeasured qudits are randomly compiled to make the effective noise on the unmeasured register follow a stochastic model.
We propose a similar implementation, where when qudits are idle, we apply first and last a random stabilizer, and in between we apply a refocusing sequence that persists for the remaining duration of the idle time.
There are many choices of refocusing techniques; one should be chosen to optimize performance based on the noise model that is present in the system and possibly on the encoding.
This process should be applied anywhere that qudits would otherwise idle in a computation and not just when measurements are in progress.
Because the ideal implementation of an idling sequence is the identity, twirling gates can and should be added to constrain the effective noise map to a stochastic Weyl channel.

\subsection{Code switching}\label{sec:codeSwitch}

While we do not look at code switching methods explicitly in this paper, it should be noted that computations which involve code switching can be randomly compiled by applying randomizations according to whatever the encoding is at any given step in a circuit.
For example, in switching from a code $C$ to a code $C'$, the minimal example would be to apply random stabilizers from code $C$ prior to switching and random stabilizers from code $C'$ after switching.
This ensures that any portion of the state outside of the codespace prior to the switch does not get mapped into the new code in a way that is disadvantageous and ensures that the initial state after switching to $C'$ is in the syndrome space of $C'$.
For any steps where the register is unencoded, randomized compiling should be used, and for any where the system is in a concatenated encoding, stabilizers (and logical operations) of the combined code should be used.

\subsection{Sampling Algorithms}\label{sec:singleShot}

Some quantum computing algorithms require only sampling from a distribution rather than estimating the full distribution.
In this section, we will argue that for such algorithms, LRC (and RC) can still be used.

For each compilation $c$ of a circuit, we have an output probability distribution $\vec{p}_c$.
The average over random compilations results in an average over output distributions, where the average output distribution $\vec{p}_{avg}$ approximates the ideal distribution $\vec{p}_{ideal}$, that is
\begin{align}
    \vec{p}_{ideal}\approx \vec{p}_{avg} = \frac{1}{|c|}\sum_{c}\vec{p}_c.
\end{align}
The probability of getting a ditstring $b$ when sampling a single shot from $\vec{p}_{avg}$ is given by
\begin{align}
    \vec{p}_{avg}(b) = \frac{1}{|\{c\}|}\sum_c \vec{p}_c(b)
    \label{eq:sampleFrAvg}
\end{align}
because the distributions $\vec{p}_c$ are mutually independent.
If we take a single shot with a single compilation sampled uniformly at random, the probability of returning ditstring $b$ is
\begin{align}
    \sum_c P(\vec{p}_c)\vec{p}_c(b) = \frac{1}{|\{c\}|}\sum_c \vec{p}_c(b),
\end{align}
where the equality comes from the different compilations being equiprobable (i.e. $P(\vec{p}_c) = 1/|\{c\}|$).
This is the same probability as in \cref{eq:sampleFrAvg}, so we can conclude that taking a single shot with a randomly sampled compilation is equivalent to taking a single shot from the average distribution which approximates the ideal distribution.

\section{Conclusion}

Noise in quantum error corrected systems often deviates from the common assumption of a stochastic Weyl model with ideal or probabilistically misreported syndrome measurements.
In this paper, we proposed a method that generalizes and combines existing compilation strategies in the context of fault-tolerant quantum computations to show how noise can be tailored to follow the desired form that is generally assumed.
Under our protocol, noisy quantum error corrected implementations are easier to model and the remaining noise is easier to correct using the standard Weyl recovery methods.

\bibliographystyle{unsrt}
\bibliography{library}

\end{document}